\documentclass[pdflatex,sn-mathphys-num]{sn-jnl}

\usepackage{graphicx}%

\usepackage{multirow}%
\usepackage{amsmath,amssymb,amsfonts}%
\usepackage{amsthm}%
\usepackage{mathrsfs}%
\usepackage[title]{appendix}%
\usepackage{xcolor}%
\usepackage{textcomp}%
\usepackage{manyfoot}%
\usepackage{booktabs}%
\usepackage{algorithm}%
\usepackage{algorithmicx}%
\usepackage{algpseudocode}%
\usepackage{listings}%
\usepackage{physics}%
\usepackage{tikz}%
\theoremstyle{thmstyleone}%
\usepackage{quantikz}
\theoremstyle{thmstyletwo}%

\theoremstyle{thmstylethree}%
\newtheorem{conjecture}{Conjecture}%
\newtheorem{theorem}{Theorem}%
\newtheorem{proposition}{Proposition}%
\newtheorem{corollary}{Corollary}%
\newtheorem{lemma}{Lemma}%

\pgfdeclarelayer{background}
\pgfsetlayers{background,main}
\usetikzlibrary{calc, arrows.meta}

\newcommand{\mc}[1]{\mathcal{#1}}
\newcommand{\mb}[1]{\mathbb{#1}}
\newcommand{\xl}{\lambda}
\renewcommand{\L}{\mc{L}}

\newcommand{\id}{\mathbb{I}}

\newcommand{\ghz}[2]{\text{GHZ}_{#1,#2}}

\newcommand{\E}{\mc{E}}

\newcommand{\spin}{	
	\begin{tikzpicture}
		\filldraw[black] (0,0) circle (0.35cm);
		\draw[ line width=1.8pt, -{Triangle[length=10pt,width=8pt]}] (0,-.8) -- (0,1);
	\end{tikzpicture}
}

\newcommand{\CSN}{\mathbb{S}_\lambda}
\newcommand{\ISN}{\mathbb{S}_I}
\newcommand{\F}{F^{*}}
\newcommand{\G}{{K_N^{(k)}}}
\newcommand{\target}{\ghz{N}{d^m}}

\providecommand{\red}[1]{}
\providecommand{\blue}[1]{}
\renewcommand{\red}[1]{{\color{red}#1}}
\renewcommand{\blue}[1]{{\color{blue}#1}}

\begin{document}
	
	\title{Optimal GHZ-State Distribution in LOSR Quantum Networks via Local Decoding from Information Sets} 
	\author*[1]{\fnm{Leonardo} \sur{Oleynik}}\email{leonardo.oleynik@uni.lu}
	\equalcont{These authors contributed equally to this work.} 
	\author*[1]{\fnm{Shehbaz} \sur{Tariq}}\email{shehbaz.tariq@uni.lu}
	\equalcont{These authors contributed equally to this work.} 
	\author[1]{\fnm{Symeon} \sur{Chatzinotas}}\email{symeon.chatzinotas@uni.lu}
	\equalcont{These authors contributed equally to this work.}
	
	\affil*[1]{\orgdiv{Interdisciplinary Centre for Security, Reliability, and Trust (SnT)}, \orgname{University of Luxembourg}, \orgaddress{\postcode{L-1855}, \city{Luxembourg City}, \country{Luxembourg}}}

	\abstract{Distributing multipartite entanglement is a prerequisite for scalable quantum networks. Networks restricted to local operations and shared randomness (LOSR) avoid the quantum-memory and latency costs that real-time classical communication imposes on LOCC-based networks, but with only bipartite sources they cannot prepare usefull GHZ states. In earlier work we conjectured that multipartite sources lift this restriction, supporting the claim with a single numerical example. Here we prove the conjecture for regular and uniform networks of arbitrary size. Identifying the hyperedges of the network with the coordinates of a linear code, we show that whenever the edges incident to each node form an information set, a fixed set of local unitaries---the local decoders of the code---converts the source state into an $N$-party GHZ state with fidelity $d^{\,m-M}$ while using no classical communication, where $m$ and $M$ are the number of edges per node and in total. For the complete hypergraphs $K_N^{(N-1)}$ this fidelity equals $1/d$, and we prove that it is optimal among local-unitary strategies, surpassing the bipartite bound (for instance $1/2$ versus $1/8$ in the four-node case). Multipartite sources together with shared randomness can therefore replace real-time communication for entanglement distribution.}
		
	
	
	\maketitle
	\section{\label{sec: intro}Introduction}
	
	Quantum networks (QNs) are expected to extend quantum information processing across distant sites, ultimately enabling secure distributed computation between remote parties~\cite{Qinternet,Kimble2008,Wehner2018}. This prospect has made them a central theme of recent theoretical and experimental work. At the heart of the programme lies the generation and distribution of multipartite entangled resources---graph and GHZ states in particular---which underpin tasks such as quantum key distribution~\cite{48QKD,QKDexp}, blind quantum computation~\cite{BQC7}, and quantum metrology~\cite{Qmetro1,Qmetro2,Qmetro3}. Although such states have been realized in small or carefully controlled settings~\cite{12control,16control}, distributing them over large distances and across many parties remains difficult, mainly because the precision required for coherent control grows steeply with the number of systems involved.
	
	What a network can distribute is dictated by the resources its nodes share. In the most elementary case, nodes hold independent local states and coordinate only through classical communication (CC), as sketched in the leftmost of Fig.~\ref{fig: networks}. Such purely classical networks are simple to operate but cannot create multipartite entanglement, and therefore offer no quantum advantage. Supplying neighbouring nodes with bipartite entanglement in addition to CC changes the picture entirely: with sufficiently entangled bipartite resources, local operations and classical communication (LOCC) can in principle generate any multipartite state, GHZ states included. The price is operational---LOCC protocols rely on quantum memories that store each system coherently until the necessary classical messages are exchanged. The resulting latency is negligible in the laboratory but becomes significant once a network reaches intercontinental scales.

	This motivates networks governed by local operations and shared randomness (LOSR), a low-latency middle ground between the classical and LOCC regimes. Here neighbouring nodes still share bipartite entanglement, but rather than exchanging messages in real time they rely on a classical variable distributed beforehand (Fig.~\ref{fig: networks}); once the quantum systems are in place, no further communication occurs. Such common randomness, together with a shared time reference, is already distributed at intercontinental scale by Global navigation satellite system broadcasts, so this prior coordination needs no real-time link.%
	Removing real-time CC eliminates the latency problem, yet it also constrains the network's ability to generate multiparitite entanglement, and the severity of that constraint is far from obvious. A sequence of recent results has progressively tightened the limits on multipartite entanglement in LOSR networks~\cite{20GNME,21EwUni,21GNME,22QNandS,23GNME,24Otfried,25Otfried-clone}, to the point of questioning whether any meaningful advantage survives at all~\cite{25Otfried}. The sharpest of these shows that, with bipartite entangled sources and more than three parties, LOSR networks cannot prepare GHZ states any better than classical networks can~\cite{25Otfried}---a clear no-go for entanglement distribution in QNs with bipartite resources and no real-time communication.
	
	In a recent conference paper~\cite{11500272} we proposed a route around this obstruction: increasing the shared resources from bipartite to multipartite entanglement. There we conjectured that complete $k$-uniform hypergraph networks with multipartite sources can surpass the bipartite bound, and we supported the conjecture with a single instance --- a four-node network with tripartite sources, for which numerical optimization returned a GHZ fidelity of $1/2$ against the bipartite value $1/8$, realized by a heuristically chosen set of \textsc{swap} and \textsc{cnot} unitaries. That evidence was limited in three respects: it concerned only $N=4$, the optimality of $1/2$ was inferred from numerics rather than proved, and the local unitaries were chosen \emph{ad hoc}. The present work removes all three limitations. We prove the conjecture for any regular and uniform hypergraph network; we replace the \emph{ad hoc} unitaries by a systematic construction in which the hyperedges of the network are identified with the coordinates of a linear code and the required local unitaries are the code's local decoders; and we prove analytically that, for complete hypergraphs $K_N^{(N-1)}$, the resulting fidelity $1/d$ is optimal and achievable by our systematic construction.
	
	The resource produced this way is operationally meaningful: although entanglement distillation may be unavailable in the LOSR paradigm~\cite{LOSRnew}, the entanglement of noisy multipartite states can still be localized~\cite{04ent_local,05ent_local} or concentrated~\cite{13ent_cluster}. From a resource theory perspective, our results show that shared randomness combined with multipartite sources can be traded for real-time CC in the distribution of multipartite entanglement, pointing to a concrete resource with which low-latency networks may serve as a practical alternative to fully fledged LOCC-based ones.
	
	The remainder of the paper is organized as follows. Section~\ref{sec: ps} sets up hypergraph LOSR networks and the fidelity-based figure of merit we use to measure the multipartite entanglement they can generate. Section~\ref{sec: results} presents our results: the information-set construction and its fidelity (Sec.~\ref{subsec: construction}), the proof of advantage and optimality for the complete networks $K_N^{(N-1)}$ (Sec.~\ref{subsec: advantage}), a numerical local-unitary benchmark---Riemannian optimization over the product-unitary manifold---that confirms optimality for $K_N^{(N-1)}$ and corroborates the construction at other source arities (Sec.~\ref{subsec:lu_optimization}), and the explicit four-node example that recovers and explains our earlier findings (Sec.~\ref{subsec: example}). Section~\ref{sec: conc} concludes. Detailed proofs are collected in the Appendix.
	
	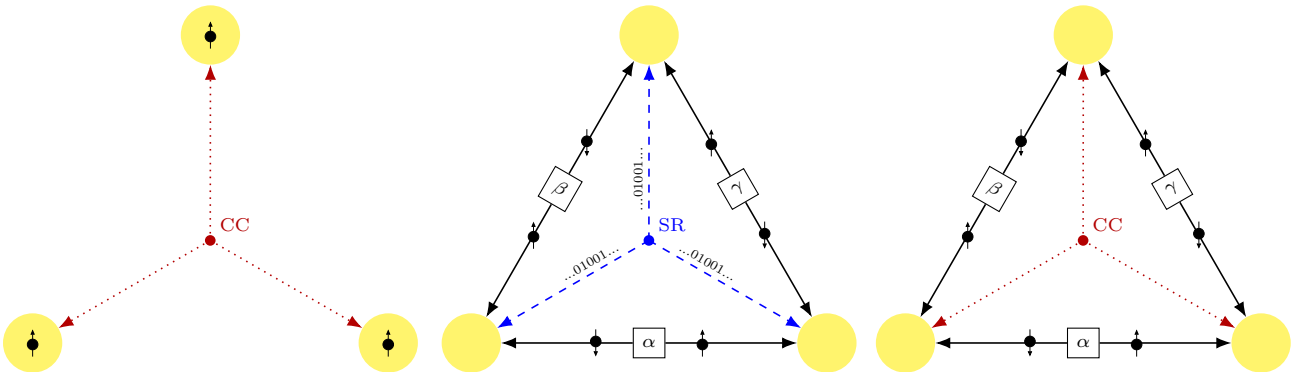
\begin{figure*}[htbp]
		\centering
		\begin{tikzpicture}
			\node (a) at (0,0) {	\begin{tikzpicture}[
		font=\footnotesize, 
		text centered,
		line_in/.style={-{Latex},semithick,dotted,red!70!black},
		vertex2/.style={circle, fill=yellow!70, minimum size=2*\r,inner sep=0pt},
		source/.style={draw, shape=rectangle, minimum size=\r}
		]
		\pgfmathsetlengthmacro{\L}{\textwidth/3.7} 
		\pgfmathsetlengthmacro{\r}{\L/12}

		\coordinate (C) at (0,0);
		\coordinate (B) at (\L,0);
		\coordinate (A) at (60:\L); 
		\coordinate (G) at ($ (C)!.333!(B) + (C)!.333!(A) $);

		\coordinate (c1) at (60:\r);
		\coordinate (c2) at (30:\r);
		\coordinate (c3) at (\r,0);

		\coordinate (a1) at ($ (A)+ (-60:\r)$);		
		\coordinate (a2) at ($ (A)+ (0,-\r)$);
		\coordinate (a3) at ($ (A)+ (-90-30:\r)$);
		
		\coordinate (b1) at ($ (B)+ (-\r,0)$);		
		\coordinate (b2) at ($ (B)+ (90+60:\r)$);
		\coordinate (b3) at ($ (B)+ (90+30:\r)$);

		\fill[red!70!black] (G) circle (2pt) node[above right] {CC};
		
		


		\draw[line_in] (G) -- (a2);
		\draw[line_in] (G) -- (b2);
		\draw[line_in] (G) -- (c2);
		

		\node[vertex2] at (A) {};
		\node[scale=.2] at (A) { \spin};		
		\node[vertex2] at (B) {};		
		\node[scale=.2] at (B) { \spin};		
		\node[vertex2] at (C) {};		
		\node[scale=.2] at (C) { \spin};		
		

	\end{tikzpicture} };
			\node at (.33\textwidth,0){			\begin{tikzpicture}[
		font=\footnotesize, 
		text centered,
		line_in/.style={-{Latex},semithick,dashed, blue},
		line_e/.style={-{Latex},semithick},
		vertex2/.style={circle, fill=yellow!70, minimum size=2*\r,inner sep=0pt},
		source/.style={draw, shape=rectangle, minimum size=\r}
		]
		\pgfmathsetlengthmacro{\L}{\textwidth/3.7}  
		\pgfmathsetlengthmacro{\r}{\L/12}

		\coordinate (C) at (0,0);
		\coordinate (B) at (\L,0);
		\coordinate (A) at (60:\L); 
		\coordinate (G) at ($ (C)!.333!(B) + (C)!.333!(A) $);

		\coordinate (c1) at (60:\r);
		\coordinate (c2) at (30:\r);
		\coordinate (c3) at (\r,0);

		\coordinate (a1) at ($ (A)+ (-60:\r)$);		
		\coordinate (a2) at ($ (A)+ (0,-\r)$);
		\coordinate (a3) at ($ (A)+ (-90-30:\r)$);
		
		\coordinate (b1) at ($ (B)+ (-\r,0)$);		
		\coordinate (b2) at ($ (B)+ (90+60:\r)$);
		\coordinate (b3) at ($ (B)+ (90+30:\r)$);

		\fill[blue] (G) circle (2pt) node[above right] {SR};
		
		\node[source, rotate=30] (a) at ($ (A)!0.5!(B) $) {$\gamma$};
		\node[source					 ] (b) at ($ (B)!0.5!(C) $) {$\alpha$};
		\node[source, rotate=-30] (c) at ($ (A)!0.5!(C) $) {$\beta$};
		
		\node[scale=.2] at ($(a)!0.3!(A)$) {\spin};
		\node[scale=.2,rotate=180] at ($(a)!0.3!(B)$) {\spin};
		\node[scale=.2] at ($(b)!0.3!(B)$) {\spin};
		\node[scale=.2,rotate=180] at ($(b)!0.3!(C)$) {\spin};
		\node[scale=.2] at ($(c)!0.3!(C)$) {\spin};
		\node[scale=.2, rotate=180] at ($(c)!0.3!(A)$) {\spin};

		\node[rotate=90,scale=.7,above] (ga) at ($ (G)!0.3!(A)$) {$...01001...$};
		\node[rotate=-30,scale=.7,above] (gb) at ($ (G)!0.3!(B)$) {$...01001...$};
		\node[rotate=30,scale=.7,above] (gc) at ($ (G)!0.3!(C)$) {$...01001...$};
	
		\draw[line_in] (G) -- (a2);
		\draw[line_in] (G) -- (b2);
		\draw[line_in] (G) -- (c2);
		
		\draw[line_e] (a) -- (a1);
		\draw[line_e] (a) -- (b3);
		\draw[line_e] (b) -- (b1);
		\draw[line_e] (b) -- (c3);
		\draw[line_e] (c) -- (c1);
		\draw[line_e] (c) -- (a3);

		\node[vertex2] at (A) {};
		\node[vertex2] at (B) {};		
		\node[vertex2] at (C) {};		

	\end{tikzpicture}};
			\node at (.66\textwidth,0) {			\begin{tikzpicture}[
		font=\footnotesize, 
		text centered,
		line_in/.style={-{Latex},semithick,dotted, red!70!black},
		line_e/.style={-{Latex},semithick},
		vertex2/.style={circle, fill=yellow!70, minimum size=2*\r,inner sep=0pt},
		source/.style={draw, shape=rectangle, minimum size=\r}
		]
		\pgfmathsetlengthmacro{\L}{\textwidth/3.7}  
		\pgfmathsetlengthmacro{\r}{\L/12}

		\coordinate (C) at (0,0);
		\coordinate (B) at (\L,0);
		\coordinate (A) at (60:\L); 
		\coordinate (G) at ($ (C)!.333!(B) + (C)!.333!(A) $);
		
		\coordinate (c1) at (60:\r);
		\coordinate (c2) at (30:\r);
		\coordinate (c3) at (\r,0);
		
		\coordinate (a1) at ($ (A)+ (-60:\r)$);		
		\coordinate (a2) at ($ (A)+ (0,-\r)$);
		\coordinate (a3) at ($ (A)+ (-90-30:\r)$);
		
		\coordinate (b1) at ($ (B)+ (-\r,0)$);		
		\coordinate (b2) at ($ (B)+ (90+60:\r)$);
		\coordinate (b3) at ($ (B)+ (90+30:\r)$);

		\fill[red!70!black] (G) circle (2pt) node[above right] {CC};
		
		\node[source, rotate=30] (a) at ($ (A)!0.5!(B) $) {$\gamma$};
		\node[source					 ] (b) at ($ (B)!0.5!(C) $) {$\alpha$};
		\node[source, rotate=-30] (c) at ($ (A)!0.5!(C) $) {$\beta$};
		
		\node[scale=.2] at ($(a)!0.3!(A)$) {\spin};
		\node[scale=.2,rotate=180] at ($(a)!0.3!(B)$) {\spin};
		\node[scale=.2] at ($(b)!0.3!(B)$) {\spin};
		\node[scale=.2,rotate=180] at ($(b)!0.3!(C)$) {\spin};
		\node[scale=.2] at ($(c)!0.3!(C)$) {\spin};
		\node[scale=.2, rotate=180] at ($(c)!0.3!(A)$) {\spin};

		
		\draw[line_in] (G) -- (a2);
		\draw[line_in] (G) -- (b2);
		\draw[line_in] (G) -- (c2);
		
		\draw[line_e] (a) -- (a1);
		\draw[line_e] (a) -- (b3);
		\draw[line_e] (b) -- (b1);
		\draw[line_e] (b) -- (c3);
		\draw[line_e] (c) -- (c1);
		\draw[line_e] (c) -- (a3);

		\node[vertex2] at (A) {};
		\node[vertex2] at (B) {};		
		\node[vertex2] at (C) {};		

	\end{tikzpicture}};
		\end{tikzpicture}
		\caption{\label{fig: networks}From left to right: classical, LOSR, and LOCC triangle networks. On the left, each node (yellow circle) holds a separate quantum state and may communicate over a classical channel (dotted line). In the middle, each node holds one share of each of the two bipartite sources on its adjacent edges, e.g.\ $\beta$ and $\gamma$, and coordinates its local operations through a classical distribution shared \emph{a priori} (dashed line). On the right, nodes share bipartite entanglement and communicate classically.}
	\end{figure*}

	\section{\label{sec: ps} Entanglement in LOSR Quantum Networks}
	
	We start from the smallest LOSR network that is not trivial. Following~\cite{22QNandS}, let Alice, Bob, and Charlie attempt to synthesize a tripartite state $\sigma$ using only local processing of three bipartite sources $\sigma_{ab'},\sigma_{bc'}$ and $\sigma_{ca'}$ of arbitrary dimension $d\times d$, the two shares of each source being sent to different nodes. Every node applies a local map $\E_X^\xl$ selected by a shared classical variable $\xl$, so the reachable states are
	\begin{equation}
		\sigma = \sum_{\xl} p_\xl     \E_A^\xl \otimes \E_B^\xl \otimes \E_C^\xl    \pqty{\sigma_{ca'}\otimes\sigma_{bc'}\otimes\sigma_{ab'}}.
		\label{eq:def-rho-tri}
	\end{equation}
	The extreme points of this set are the independent triangle network states~\cite{21EwUni}
	\begin{equation}
		\sigma^\xl =      \E_A^\xl \otimes \E_B^\xl \otimes \E_C^\xl    \pqty{\sigma_{ca'}\otimes\sigma_{bc'}\otimes\sigma_{ab'}},
	\end{equation}
	and a general feasible state is a convex mixture of them. The maps $\E_X^\xl$ may be any trace-preserving positive maps, hence describe arbitrary channels, discarding included~\cite{20GNME}. 
	
	To accommodate sources shared by more than two nodes, we describe the network by an undirected hypergraph $G=(V,E)$, with $V$ and $E$ the sets of vertices and edges. A $k$-partite source is a $k$-edge $e=\{i_1,\dots,i_k\}\in E$ joining its $k$ nodes; the nodes are the elements of $V$, and $N=\abs{V}$. As before, each node $i\in V$ applies a local map $\E_i^{\xl}$ coordinated by $\xl$, now acting on the $d^{m_i}$-dimensional system formed by its $m_i=\abs{E(i)}$ incident sources, where $E(i)=\{e\in E\,\mid\, i\in e\}$. Any correlation the network can realize then has the form
	\begin{equation}
		\rho_G = \sum_{\xl} p^\xl    \pqty{\otimes_{i\in V} \E_i^{\xl}}  \pqty{\otimes_{e\in E} \rho_e},
		\label{eq: def-rho}
	\end{equation}
	with $p^\xl$ a probability distribution over $\xl$, and its extreme points are the independent hypergraph network states
	\begin{equation}
		\rho^\xl = \pqty{\otimes_{i\in V} \E_i^{\xl}}  \pqty{\otimes_{e\in E} \rho_e}.
		\label{eq: def-in}
	\end{equation}
	We write $\CSN$ and $\ISN$ for the sets of states of the forms \eqref{eq: def-rho} and \eqref{eq: def-in}, respectively. A $k$-uniform hypergraph is \emph{$m$-regular} when every vertex is incident to the same number of edges, $m_i=m$ for all $i\in V$; the common degree $m$ then fixes the local dimension $d^{m}$ at every node. A \emph{complete} $k$-uniform hypergraph---one in which every $k$-subset of $V$ is an edge---is the special case in which each vertex meets all $\binom{N-1}{k-1}$ edges that contain it; it is thus $m$-regular, with $m=\binom{N-1}{k-1}$ and $M=\abs{E}=\binom{N}{k}$ edges in total, and we denote it $\G$. 
	
	To quantify the multipartite entanglement such a network can deliver, we use the largest fidelity attainable with the target $N$-party GHZ state of local dimension $d^m$, $\ket{\ghz{N}{d^m}}=d^{-m/2}\sum_{i=0}^{d^m-1}\ket{i}^{\otimes N}$:
	\begin{equation}
		\F_G= \max_{\rho_G \in\CSN} F(\rho_G,\ket{\target}).
		\label{eq: F-star}
	\end{equation}
	Because $F$ is convex in $\rho_G$ over $\CSN$, the maximum is attained at an extreme point, so it suffices to optimize over the independent states $\ISN$:
	\begin{equation}
		\F_G = \max_{\rho^\xl\in\ISN} F(\rho^\xl,\ket{\target}).
		\label{eq:bound}
	\end{equation}
	
%
	
	Intuitively, sharing multipartite entanglement among several nodes at once should let a network distribute more of it than sharing only bipartite entanglement between neighbours, the more so as multipartite entanglement is structurally far richer than bipartite entanglement~\cite{entangComplex}. We therefore expect $k$-uniform networks to outperform their $2$-uniform counterparts, the latter being upper bounded by $1/d^m$ ~\cite{25Otfried}, where $d^{m}$ denotes the dimension of the target GHZ state. We first stated this expectation in~\cite{11500272}, where it was verified only for a square network associated to a complete $3$-uniform hypergraph; we generalize it here as the claim to be proved.
	
	
	\begin{conjecture}\label{conj:k-partite}
		($k$-partite advantage): $\forall N$ and $2\leq k$ there is a $m$-regular $k$-uniform hypergraph $G=\G$ such that $\F_{K_N^{(2)}} \leq 1/d^{m} \leq \F_\G$.
	\end{conjecture}

	\section{\label{sec: results} Results}
	
	Throughout this section---both for the construction and for the optimality proof---we restrict to local unitaries $U_i$ acting on pure, maximally entangled sources. This is the natural setting in which to exhibit advantage: general local channels only enlarge the feasible set, so any fidelity reached with unitaries is a valid lower bound on $\F_G$~\cite{25Otfried}. Our results are twofold. We first give a construction (Sec.~\ref{subsec: construction}) that, for any network whose incident edges form information sets, produces explicit local unitaries reaching fidelity $d^{\,m-M}$ with the GHZ target. We then prove (Sec.~\ref{subsec: advantage}) that for the complete hypergraphs $K_N^{(N-1)}$ this value equals $1/d$ and is optimal, which simultaneously establishes Conjecture~\ref{conj:k-partite} for every $N$ and shows that the construction cannot be improved. Section~\ref{subsec: example} works out the four-node case explicitly.
	
	\subsection{Information-set construction}
	\label{subsec: construction}
	
	We present a constructive local-unitary ansatz that turns a hypergraph source state into a high-dimensional $N$-party GHZ state with maximum fidelity, requiring no classical communication. The construction rests on a linear code whose coordinates are identified with the hyperedges of the network: each node should be able to reconstruct one and the same global message from only the source labels available to it, which is precisely the role of an \emph{information set} in coding theory. We assume throughout that every node meets the same number of edges, $m_i=m\;\forall i$ (an \emph{$m$-regular} hypergraph), as is the case for complete $k$-uniform hypergraphs.
	
	Identify the $M$ hyperedges of $G$ with the coordinates of a linear code over $\mb{F}_d$,
	\begin{equation}
		\mc{C} = \left\{\, \vb{x} = G\vb{z} \,:\, \vb{z}\in\mb{F}_d^{m} \,\right\} \subseteq \mb{F}_d^{M},
		\label{eq:info-code}
	\end{equation}
	generated by a full-column-rank matrix $G\in\mb{F}_d^{M\times m}$; $\vb{z}$ is the message and $\vb{x}=G\vb{z}$ the codeword carried by the edges. For node $i$, let $G_i=G_{E(i),:}\in\mb{F}_d^{m\times m}$ be the submatrix of $G$ formed by the rows indexed by the incident edges $E(i)$, and let $\vb{x}_{E(i)}$ denote the subvector of $\vb{x}$ on those coordinates. The construction applies whenever every incident submatrix is invertible,
	\begin{equation}
		\operatorname{rank}_{\mb{F}_d}(G_i)=m, \qquad \forall i\in V,
		\label{eq:info-set-condition}
	\end{equation}
	i.e.\ the coordinates of each node form an information set of $\mc{C}$. On the codeword sector $\vb{x}=G\vb{z}$ a node then observes $\vb{x}_{E(i)}=G_i\vb{z}$ and recovers the message by applying the local decoder
	\begin{equation}
		D_i = G_i^{-1}, \qquad D_i\,\vb{x}_{E(i)} = G_i^{-1}G_i\vb{z} = \vb{z}.
		\label{eq:decoder}
	\end{equation}
	Each $D_i$ is an invertible $\mb{F}_d$-linear map, hence a permutation of the computational basis, and therefore induces a local permutation unitary on the $d^m$-dimensional space of node $i$,
	\begin{equation}
		U_i\ket{\vb{y}}_i = \ket{D_i\vb{y}}_i, \qquad \vb{y}\in\mb{F}_d^{m}.
		\label{eq:U-i}
	\end{equation}
	The protocol applies the product unitary
	\begin{equation}
		U_{\mathrm{IS}} = \bigotimes_{i\in V}U_i .
		\label{eq:U-IS}
	\end{equation}
	Since the $U_i$ are fixed in advance and act only on local registers, no classical communication is needed at runtime.
	
	Applied to the state generated by maximally entangled $\ket{\ghz{k}{d}}$ sources, the construction reaches a fidelity that depends only on the two code parameters $m$ and $M$ (Appendix~\ref{app: fidelity}):
	\begin{proposition}\label{prop:fidelity}
		Let $\ket{\phi}$ be the trivial network state\footnote{The \emph{trivial network state} is the independent network state~\eqref{eq: def-in} in which every local map $\E_i^{\xl}$ is the identity, i.e.\ the state before any operation is applied.} of a regular $k$-uniform hypergraph with $\ket{\ghz{k}{d}}$ sources, and let the information-set condition \eqref{eq:info-set-condition} hold. Then
		\begin{equation}
			F_{\mathrm{IS}} = F\!\left(U_{\mathrm{IS}}\ket{\phi},\,\ket{\ghz{N}{d^m}}\right) = d^{\,m-M}.
			\label{eq:F-IS}
		\end{equation}
	\end{proposition}
	The benchmark to surpass is the product (classical) value $d^{-m}$, which bipartite-source LOSR networks cannot exceed~\cite{25Otfried}. The construction therefore yields a quantum advantage precisely when $d^{\,m-M}>d^{-m}$, that is when
	\begin{equation}
		M<2m \quad\Longleftrightarrow\quad \binom{N}{k}<2\binom{N-1}{k-1} \quad\Longleftrightarrow\quad N<2k .
		\label{eq:advantage-condition}
	\end{equation}
	Hence, whenever a complete $k$-uniform network satisfies $N<2k$, multipartite sources outperform every bipartite-source network on the same task.
	
	\subsection{Advantage and optimality for the complete networks \texorpdfstring{$K_N^{(N-1)}$}{K	extsubscript{N}	extasciicircum(N-1)}}
	\label{subsec: advantage}
	
	The boundary case $k=N-1$ (equivalently $N=k+1$) is the most symmetric and gives the cleanest statement. Every $(N-1)$-subset of $V$ is then an edge, so there are $M=\binom{N}{N-1}=N$ edges, each node meets $m=\binom{N-1}{N-2}=N-1$ of them, and the target is $\ket{\ghz{N}{d^{N-1}}}$. The fidelity \eqref{eq:F-IS} becomes
	\begin{equation}
		F_{\mathrm{IS}} = d^{\,(N-1)-N} = \frac{1}{d}.
		\label{eq:F-NN1}
	\end{equation}
	Observe that the complete graph $K_N^{(2)}$ on the same $N$ nodes also has $m=N-1$ and the \emph{same} target $\ket{\ghz{N}{d^{N-1}}}$, yet \cite{25Otfried} bounds it above by $1/d^{N-1}$, whereas multipartite sources reach $1/d$. We now show that this value is optimal.
	\begin{theorem}\label{thm:optimality}
		For $N\geq 4$, the complete hypergraph $K_N^{(N-1)}$ with $\ket{\ghz{N-1}{d}}$ sources satisfies
		\begin{equation}
			\F_{K_N^{(N-1)}} = \max_{U=\otimes_i U_i}\abs{\mel{\ghz{N}{d^{N-1}}}{U}{\phi}}^2 = \frac{1}{d},
		\end{equation}
		the maximum being over all local unitaries and attained by $U_{\mathrm{IS}}$.
	\end{theorem}
	The proof is given in Appendix~\ref{app: optimality}: achievability is the construction of Sec.~\ref{subsec: construction}, while the upper bound follows by relaxing to a balanced bipartition and using that both the GHZ target and the GHZ-source state have uniform Schmidt spectra across it. This establishes optimality within local-unitary strategies on maximally entangled sources, which is the regime in which the advantage is demonstrated; general local channels enlarge the feasible set and can only raise $\F_G$~\cite{25Otfried}, so whether they exceed $1/d$ is left open.
	
	Because $1/d\geq 1/d^{N-1}$ for every $N$ (strictly for $N\geq 3$), choosing $k=N-1$ gives $\F_{K_N^{(2)}}\leq 1/d^{N-1}\leq\F_{K_N^{(N-1)}}$, which proves the conjecture in full generality.
	\begin{corollary}\label{cor:conjecture}
		Conjecture~\ref{conj:k-partite} holds for every $N$, witnessed by the complete hypergraph $K_N^{(N-1)}$.
	\end{corollary}

    \subsection{Numerical local-unitary benchmark}
\label{subsec:lu_optimization}

We further benchmark the information-set construction by directly optimizing over
the same local-unitary class used above. Let $D=d^m$ denote the local Hilbert-space
dimension at each node. The local-unitary fidelity is defined as
\begin{align}
F_G^{\mathrm{LU}}
=
\max_{\{U_i\}_{i\in V}}
\abs{
\mel{\target}
{
\bigotimes_{i\in V} U_i
}
{\phi}
}^{2},
\qquad
U_i\in \mathrm{U}(D),
\label{eq:lu_optimization}
\end{align}
where $\mathrm{U}(D)$ is the group of $D\times D$ unitary matrices. Since local
unitaries form a subclass of the local operations allowed in LOSR networks,
$F_G^{\mathrm{LU}}\leq F_G^\ast$. The information-set decoder gives a feasible
point of \eqref{eq:lu_optimization}, and therefore
\begin{align}
F_G^{\mathrm{LU}}
\geq
F_{\mathrm{IS}}
=
d^{m-M}.
\label{eq:lu_lower_bound}
\end{align}
For the complete hypergraphs $K_N^{(N-1)}$, Theorem~1 gives the matching upper
bound, so that $F_G^{\mathrm{LU}}=F_{\mathrm{IS}}=1/d$.

The optimization in \eqref{eq:lu_optimization} is nonconvex because each block
must satisfy $U_i^\dagger U_i=\id_D$. We therefore minimize the negative fidelity
\begin{align}
\mc{L}(U_1,\ldots,U_N)
=
-
\abs{
\mel{\target}
{
\bigotimes_{i\in V} U_i
}
{\phi}
}^{2}
\label{eq:lu_loss}
\end{align}
over the product unitary manifold
\begin{align}
\mc{M}_{\mathrm{LU}}
=
\mathrm{U}(D)^{\times N}.
\end{align}
At each iteration, the Euclidean block gradient $\Gamma_i$ is projected onto the
tangent space of $\mathrm{U}(D)$ through the skew-Hermitian matrix
\begin{align}
\Omega_i
=
\frac{1}{2}
\pqty{
U_i^\dagger \Gamma_i
-
\Gamma_i^\dagger U_i
}.
\label{eq:lu_skew_gradient}
\end{align}
The local unitary is then updated as
\begin{align}
U_i
\leftarrow
U_i\exp\pqty{-\eta\Omega_i},
\label{eq:lu_retraction}
\end{align}
where $\eta>0$ is chosen by backtracking line search. This update preserves
$U_i^\dagger U_i=\id_D$ at every iteration. The overlap in
\eqref{eq:lu_loss} is evaluated by contracting the correlated edge-label
configurations $\vb{x}\in\mathbb{F}_d^M$ directly, avoiding construction of the
dense global unitary $\bigotimes_i U_i$ and the full $D^N$-dimensional state.

Fig.~\ref{fig:lu_benchmark} shows the numerical benchmark for $d=2$. In
Fig.~\ref{fig:lu_benchmark}(a), the information-set fidelity
$F_{\mathrm{IS}}=d^{m-M}$ is plotted against the source arity $k$ for
$N=4,5,6$. The fidelity increases as the sources become more multipartite and
reaches $F_{\mathrm{IS}}=1/2$ at the boundary family $k=N-1$. This agrees with
the advantage condition $N<2k$. In Fig.~\ref{fig:lu_benchmark}(b), the
optimization histories are shown for $K_4^{(3)}$, $K_5^{(4)}$, and
$K_6^{(5)}$. In all three cases, the Riemannian search converges to the
analytical value $F_G^{\mathrm{LU}}=1/d=1/2$, consistent with Theorem~1. The
convergence becomes slower as $N$ increases because the local dimension grows as
$D=d^{N-1}$ for $K_N^{(N-1)}$.

For complete $k$-uniform hypergraphs beyond the boundary case $k=N-1$, the local-unitary optimization is used as a numerical benchmark rather than as a proof of optimality. Its convergence to $d^{m-M}$ confirms that the information-set decoder attains the same fidelity as that achieved by direct local-unitary search. However, in the absence of a matching analytical upper bound, this numerical agreement should be interpreted only as supporting evidence for the construction, not as a certificate of global optimality.

\begin{figure}[t]
    \centering
    \begin{minipage}[t]{0.48\linewidth}
        \centering
        \includegraphics[width=\linewidth]{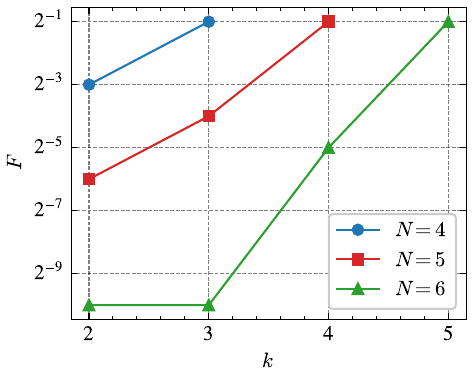}
        \vspace{1mm}
        \textbf{(a)}
    \end{minipage}
    \hfill
    \begin{minipage}[t]{0.48\linewidth}
        \centering
        \includegraphics[width=\linewidth]{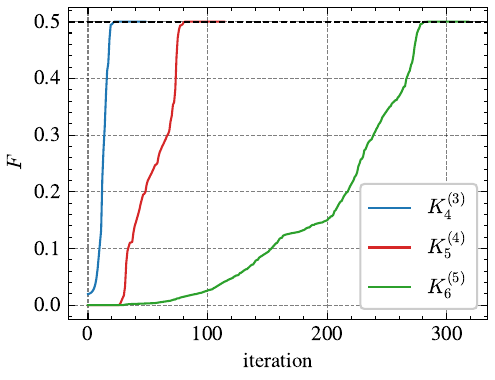}
        \vspace{1mm}
        \textbf{(b)}
    \end{minipage}
    \caption{
    Numerical local-unitary benchmark for $d=2$.
    (a) Information-set fidelity $F_{\mathrm{IS}}=d^{m-M}$ (log scale) versus the source
    arity $k$ for complete $k$-uniform hypergraphs with $N=4,5,6$. The fidelity
    reaches $F_{\mathrm{IS}}=1/2$ for the boundary family $k=N-1$.
    (b) Local-unitary optimization history for $K_4^{(3)}$, $K_5^{(4)}$, and
    $K_6^{(5)}$. The dashed line marks the analytical optimum
    $F_G^{\mathrm{LU}}=1/d=1/2$ from Theorem~1.
    }
    \label{fig:lu_benchmark}
\end{figure}

	\subsection{Explicit example: the complete network \texorpdfstring{$K_4^{(3)}$}{K4(3)}}
	\label{subsec: example}

\begin{figure}[ht]
\centering
\resizebox{\textwidth}{!}{
\begin{tabular}{cccc}
\begin{minipage}{0.24\textwidth}
    \centering
    \small
    $A:\ U_A,\quad (a,d,c)\mapsto(c,a,d)$

    \vspace{2mm}

    \begin{quantikz}[row sep=0.42cm, column sep=0.38cm]
        \lstick{$\ket{a}$} & \qw & \qw & \rstick{$\ket{c}$}\qw \\
        \lstick{$\ket{d}$} & \qw & \qw & \rstick{$\ket{a}$}\qw \\
        \lstick{$\ket{c}$} & \qw & \qw & \rstick{$\ket{d}$}\qw
    \end{quantikz}

    \vspace{1mm}
    \footnotesize
    basis relabeling
\end{minipage}
&
\begin{minipage}{0.24\textwidth}
    \centering
    \small
    $B:\ U_B,\quad (b,a,d)\mapsto(c,a,d)$

    \vspace{2mm}

    \begin{quantikz}[row sep=0.42cm, column sep=0.38cm]
        \lstick{$\ket{b}$} & \targ{}   & \targ{}   & \rstick{$\ket{c}$}\qw \\
        \lstick{$\ket{a}$} & \ctrl{-1} & \qw       & \rstick{$\ket{a}$}\qw \\
        \lstick{$\ket{d}$} & \qw       & \ctrl{-2} & \rstick{$\ket{d}$}\qw
    \end{quantikz}

    \vspace{1mm}
    \footnotesize
    $c=b\oplus a\oplus d$
\end{minipage}
&
\begin{minipage}{0.24\textwidth}
    \centering
    \small
    $C:\ U_C,\quad (c,b,a)\mapsto(c,a,d)$

    \vspace{2mm}

    \begin{quantikz}[row sep=0.42cm, column sep=0.38cm]
        \lstick{$\ket{c}$} & \ctrl{1} & \qw       & \qw       & \rstick{$\ket{c}$}\qw \\
        \lstick{$\ket{b}$} & \targ{}  & \ctrl{1}  & \targ{}   & \rstick{$\ket{a}$}\qw \\
        \lstick{$\ket{a}$} & \qw      & \targ{}   & \ctrl{-1} & \rstick{$\ket{d}$}\qw
    \end{quantikz}

    \vspace{1mm}
    \footnotesize
    $d=b\oplus a\oplus c$
\end{minipage}
&
\begin{minipage}{0.24\textwidth}
    \centering
    \small
    $D:\ U_D,\quad (d,c,b)\mapsto(c,a,d)$

    \vspace{2mm}

    \begin{quantikz}[row sep=0.42cm, column sep=0.38cm]
        \lstick{$\ket{d}$} & \targ{}   & \ctrl{1} & \ctrl{2} & \targ{}   & \rstick{$\ket{c}$}\qw \\
        \lstick{$\ket{c}$} & \qw       & \targ{}  & \qw      & \ctrl{-1} & \rstick{$\ket{a}$}\qw \\
        \lstick{$\ket{b}$} & \ctrl{-2} & \qw      & \targ{}  & \qw       & \rstick{$\ket{d}$}\qw
    \end{quantikz}

    \vspace{1mm}
    \footnotesize
    $a=b\oplus c\oplus d$
\end{minipage}

\end{tabular}
}

\caption{
Local decoder circuits for the information-set ansatz of the complete
tripartite-source network $K_4^{(3)}$. The edge labels are
$x_{ABC}=a$, $x_{ABD}=d$, $x_{ACD}=c$, and $x_{BCD}=b$, with parity relation
$b=a\oplus c\oplus d$. Each node maps its incident labels to the common
message ordering $(c,a,d)$.
}
\label{fig:k43_decoder_circuits}
\end{figure}
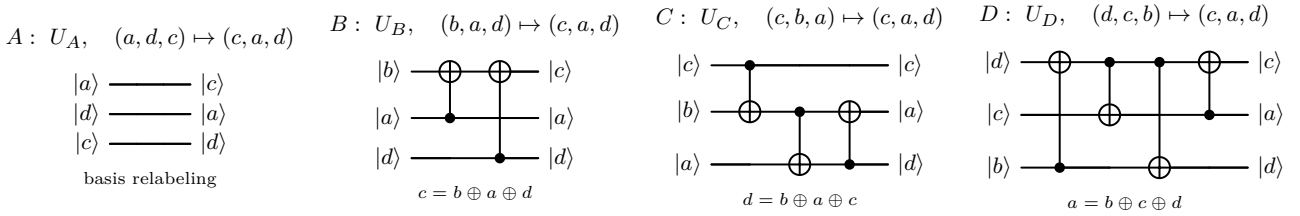

	Take $V=\{A,B,C,D\}$ with the four tripartite edges
	\begin{equation}
		E=\{ABC,\,ABD,\,ACD,\,BCD\}, \qquad M=\binom{4}{3}=4, \qquad m=\binom{3}{2}=3,
	\end{equation}
	so each node has local dimension $2^3=8$ and the target is $\ket{\ghz{4}{8}}$. Assign one binary variable to each edge,
	\begin{equation}
		x_{ABC}=a, \quad x_{ABD}=d, \quad x_{ACD}=c, \quad x_{BCD}=b,
	\end{equation}
	so that the trivial state produced by four $\ket{\ghz{3}{2}}$ sources reads
	\begin{equation}
		\ket{\phi}_{ABCD} = \frac{1}{4}\sum_{a,b,c,d=0}^{1}\ket{a,d,c}_A\ket{b,a,d}_B\ket{c,b,a}_C\ket{d,c,b}_D .
		\label{eq:raw43}
	\end{equation}
	We use the single-parity-check code defined by
	\begin{equation}
		b = a\oplus c\oplus d, \qquad \mc{C}=\left\{(a,d,c,b)\in\mb{F}_2^4: b=a\oplus c\oplus d\right\},
		\label{eq:parity43}
	\end{equation}
	a $[4,3]$ binary code. With message $\vb{z}=(c,a,d)$ and coordinate order $(ABC,ABD,ACD,BCD)$, the generator is
	\begin{equation}
		\vb{x}=\begin{bmatrix} a\\ d\\ c\\ b\end{bmatrix}=G\vb{z}, \qquad G=\begin{bmatrix} 0&1&0\\ 0&0&1\\ 1&0&0\\ 1&1&1\end{bmatrix}.
		\label{eq:G43}
	\end{equation}
	Each node sees three of the four coordinates, and any three determine $\vb{z}$, so every node's coordinates form an information set. The local decoders act on the codeword sector \eqref{eq:parity43} as
	\begin{align}
		A&:\ (a,d,c)\mapsto(c,a,d), & B&:\ (b,a,d)\mapsto(b\oplus a\oplus d,a,d)=(c,a,d),\\
		C&:\ (c,b,a)\mapsto(c,a,b\oplus a\oplus c)=(c,a,d), & D&:\ (d,c,b)\mapsto(c,b\oplus d\oplus c,d)=(c,a,d),
	\end{align}
	the equalities for $B$, $C$, $D$ using the parity relation. These define local permutation unitaries $U_A,\dots,U_D$, each reversible over $\mb{F}_2^3$ and hence a valid unitary on $\mb{C}^8$. Applying $U_{\mathrm{IS}}=U_A\otimes U_B\otimes U_C\otimes U_D$ to \eqref{eq:raw43} gives
	\begin{equation}
		\begin{split}
			\ket{\psi_{K_4^{(3)}}}_{ABCD}= \frac{1}{4}\sum_{a,b,c,d=0}^{1}&\ket{c,a,d}_A\ket{b\oplus a\oplus d,a,d}_B\\
			&\ket{c,a,b\oplus a\oplus c}_C\ket{c,b\oplus d\oplus c,d}_D .
		\end{split}
		\label{eq:psi43}
	\end{equation}
	Its overlap with $\ket{\ghz{4}{8}}=\frac{1}{\sqrt{8}}\sum_{i=0}^{7}\ket{i}^{\otimes 4}$ is nonzero only when all four three-bit labels coincide, which by \eqref{eq:psi43} happens exactly when $b=a\oplus c\oplus d$, i.e.\ for $8$ of the $16$ assignments. Hence
	\begin{equation}
		\braket{\ghz{4}{8}}{\psi_{K_4^{(3)}}} = \frac{8}{4\sqrt{8}} = \frac{1}{\sqrt{2}}, \qquad F\!\left(\ket{\psi_{K_4^{(3)}}},\ket{\ghz{4}{8}}\right)=\frac{1}{2},
		\label{eq:overlap43}
	\end{equation}
	in agreement with \eqref{eq:F-IS}, since $F_{\mathrm{IS}}=2^{m-M}=2^{3-4}=1/2$. This exceeds the product benchmark $F_{\mathrm{prod}}=1/8$ by a factor of four, providing an explicit analytical separation that uses only fixed local unitaries and no classical communication.
	
	Finally, this example recovers and explains our earlier result. The quantum circuit is illustrated in Fig.~\ref{fig:k43_decoder_circuits} and the \textsc{swap}/\textsc{cnot} unitaries reported in~\cite{11500272} are precisely the information-set decoders \eqref{eq:decoder} of the $[4,3]$ code above; the previous \emph{ad hoc} construction is thus the $k=3$ instance of the general scheme, and Theorem~\ref{thm:optimality} turns the numerically observed value $1/2$ into a proven optimum.

	\section{\label{sec: conc}Conclusion}
	
	We have shown that LOSR quantum networks equipped with multipartite sources overcome the limitations that bind their bipartite-source counterparts, and we have done so constructively and for networks of any size. Identifying the hyperedges of a complete $k$-uniform network with the coordinates of a linear code, we proved that whenever the edges incident to each node form an information set, the code's local decoders are fixed local unitaries that prepare an $N$-party GHZ state with fidelity $d^{\,m-M}$ and no classical communication. This fidelity surpasses the bipartite bound exactly when $N<2k$, and for the complete hypergraphs $K_N^{(N-1)}$ it equals $1/d$, which we proved optimal among local-unitary strategies. Choosing $k=N-1$ therefore proves the conjecture of~\cite{11500272} for every $N$, replacing its single numerical case and \emph{ad hoc} unitaries by a general, analytic, and constructive result; the four-node network is recovered as the simplest instance, with its previously conjectured fidelity $1/2$ now established as a theorem.

	Several directions remain open. The advantage condition $N<2k$ invites a systematic study of how the attainable fidelity varies with $k$ at fixed $N$ and of which regular hypergraphs are optimal; the benchmark of Sec.~\ref{subsec:lu_optimization} already maps this dependence for small networks and could guide an analytic treatment beyond the boundary family. It would also be natural to target states other than GHZ, such as cluster and W states. Our optimality proof, moreover, concerns local unitaries on maximally entangled sources; extending it to general local channels and to optimization over the source states themselves would sharpen the characterization of low-latency networks. Finally, the most effective strategy in practice is likely to involve hybrid architectures that combine the strengths of both operational regimes: employing CC (SR) between near (far) nodes can enhance global network utility while maintaining acceptable latency thresholds.
	
	\begin{appendices}
		
		\section{Fidelity of the information-set construction}
		\label{app: fidelity}
		
		We prove Proposition~\ref{prop:fidelity}. The $M$ maximally entangled sources produce the trivial network state
		\begin{equation}
			\ket{\phi} = \bigotimes_{e\in E}\ket{\ghz{k}{d}}_e = d^{-M/2}\sum_{\vb{x}\in\mb{F}_d^{M}}\;\bigotimes_{i\in V}\ket{\vb{x}_{E(i)}}_i ,
			\label{eq:phi-raw}
		\end{equation}
		where $\vb{x}_{E(i)}$ collects the labels of the $m$ edges incident to node $i$. Each $\ket{\ghz{k}{d}}_e=d^{-1/2}\sum_{x_e}\ket{x_e}^{\otimes k}$ sends one copy of $\ket{x_e}$ to each of its $k$ endpoints, so node $i$ carries $\ket{\vb{x}_{E(i)}}$; the normalization $d^{-M/2}$ follows since there are $d^M$ terms. Applying $U_{\mathrm{IS}}$ from \eqref{eq:U-IS} replaces each local label by its decoding,
		\begin{equation}
			U_{\mathrm{IS}}\ket{\phi} = d^{-M/2}\sum_{\vb{x}\in\mb{F}_d^{M}}\;\bigotimes_{i\in V}\ket{G_i^{-1}\vb{x}_{E(i)}}_i .
		\end{equation}
		With the target $\ket{\ghz{N}{d^m}}=d^{-m/2}\sum_{\vb{y}\in\mb{F}_d^{m}}\ket{\vb{y}}^{\otimes N}$, the overlap is
		\begin{equation}
			\mel{\ghz{N}{d^m}}{U_{\mathrm{IS}}}{\phi} = d^{-m/2}d^{-M/2}\sum_{\vb{y}\in\mb{F}_d^{m}}\sum_{\vb{x}\in\mb{F}_d^{M}}\;\prod_{i\in V}\braket{\vb{y}}{G_i^{-1}\vb{x}_{E(i)}} .
		\end{equation}
		The inner product $\braket{\vb{y}}{G_i^{-1}\vb{x}_{E(i)}}$ equals $1$ if $\vb{x}_{E(i)}=G_i\vb{y}$ and $0$ otherwise, so the product over $i$ is nonzero only when $\vb{x}_{E(i)}=G_i\vb{y}$ for \emph{every} node simultaneously. Since $\bigcup_i E(i)=E$, these constraints fix all coordinates of $\vb{x}$, and they are jointly consistent for exactly one $\vb{x}$ per $\vb{y}$, namely the codeword $\vb{x}=G\vb{y}$ (for which $\vb{x}_{E(i)}=(G\vb{y})_{E(i)}=G_i\vb{y}$). Conversely, the information-set condition \eqref{eq:info-set-condition} forbids any non-codeword from satisfying all constraints. Hence exactly the $d^m$ codewords contribute, each once:
		\begin{equation}
			\mel{\ghz{N}{d^m}}{U_{\mathrm{IS}}}{\phi} = d^{-m/2}d^{-M/2}\cdot d^{m} = d^{(m-M)/2},
		\end{equation}
		and $F=\abs{\mel{\ghz{N}{d^m}}{U_{\mathrm{IS}}}{\phi}}^2 = d^{\,m-M}$, as claimed. \qed
		
		\section{Proof of optimality (Theorem~\ref{thm:optimality})}
		\label{app: optimality}
		
		For the complete hypergraph $K_N^{(N-1)}$ with $N\geq 4$, the $M=N$ edges are $e_j=V\setminus\{j\}$, and node $i$ holds one qudit from every edge $e_j$ with $j\neq i$. Writing $\vb{x}_{i}=(x_j)_{j\neq i}\in\mb{F}_d^{N-1}$, the trivial state reads
		\begin{equation}
			\ket{\phi} = d^{-N/2}\sum_{\vb{x}\in\mb{F}_d^{N}}\;\bigotimes_{i\in V}\ket{\vb{x}_{i}}_i ,
			\label{eq:phi-NN1}
		\end{equation}
		and the target is $\ket{\ghz{N}{d^{N-1}}}$. The proof bounds $\F_{K_N^{(N-1)}}$ from above by $1/d$; with achievability this gives equality.
		
		\medskip\noindent\emph{Achievability.} By Proposition~\ref{prop:fidelity}, $U_{\mathrm{IS}}$ attains $F_{\mathrm{IS}}=d^{\,m-M}=d^{(N-1)-N}=1/d$, hence $\F_{K_N^{(N-1)}}\geq 1/d$. 
		
		\medskip\noindent\emph{Upper bound.} The argument rests on the maximal overlap of two bipartite pure states under local unitaries, a consequence of von Neumann's trace inequality~\cite{vonNeumann1937,Mirsky1975}.
		\begin{lemma}\label{lem:LU}
			Let $\ket{\alpha},\ket{\beta}\in\mc{H}_X\otimes\mc{H}_Y$ have squared Schmidt coefficients $\{q_a\}$ and $\{p_a\}$ (decreasing, zero-padded), and let $V_X,V_Y$ be unitaries on $\mc{H}_X,\mc{H}_Y$. Then
			\begin{equation}
				\max_{V_X,V_Y}\abs{\bra{\alpha}(V_X\otimes V_Y)\ket{\beta}} = \sum_a\sqrt{q_a\,p_a}.
				\label{eq:LU-lemma}
			\end{equation}
		\end{lemma}
		\begin{proof}
			Let $A$ and $B$ be the coefficient matrices of $\ket{\alpha}$ and $\ket{\beta}$ in the product basis; their singular values are the respective Schmidt coefficients, $\sigma_a(A)=\sqrt{q_a}$ and $\sigma_a(B)=\sqrt{p_a}$. The local unitaries act on the coefficient matrix as, $(V_X\otimes V_Y)\ket{\beta}\leftrightarrow V_X B V_Y^{T}$, and the overlap equals the Hilbert--Schmidt inner product of coefficient matrices, $\braket{\alpha}{\psi}=\operatorname{Tr}(A^{\dagger}M)$ for $\ket{\psi}\leftrightarrow M$. Hence
			\begin{equation}
				\bra{\alpha}(V_X\otimes V_Y)\ket{\beta}=\operatorname{Tr}\!\big(A^{\dagger}V_X B V_Y^{T}\big).
			\end{equation}
			Unitaries preserve singular values, so $\sigma_a(V_X B V_Y^{T})=\sigma_a(B)=\sqrt{p_a}$, and von Neumann's trace inequality gives $\abs{\operatorname{Tr}(A^{\dagger}V_X B V_Y^{T})}\leq\sum_a\sigma_a(A)\sigma_a(B)=\sum_a\sqrt{q_a p_a}$ for all $V_X,V_Y$.
		\end{proof}
		
		Fix a bipartition $V=S\cup S^{c}$ ($S\cap S^{c}=\varnothing$) with $2\leq\abs{S}\leq N-2$, which exists because $N\geq 4$. Every product of single-node unitaries is of the form $U_S\otimes U_{S^c}$, a more general case than $\otimes_{i\in V}U_i$, so considering them only increases the maximum. As the overlap is nonnegative and squaring is increasing on $[0,\infty)$, the maximization commutes with squaring, and
		\begin{equation}
			\F_{K_N^{(N-1)}} \leq \max_{U_S,U_{S^c}}\abs{\bra{\ghz{N}{d^{N-1}}}(U_S\otimes U_{S^c})\ket{\phi}}^2 = \Big(\max_{U_S,U_{S^c}}\abs{\bra{\ghz{N}{d^{N-1}}}(U_S\otimes U_{S^c})\ket{\phi}}\Big)^2 .
			\label{eq:relax}
		\end{equation}
		We read off the two Schmidt spectra across the cut. Writing $\ket{\ghz{N}{d^{N-1}}}=d^{-(N-1)/2}\sum_{\vb{y}}\ket{\vb{y}}_S^{\otimes\abs{S}}\otimes\ket{\vb{y}}_{S^c}^{\otimes\abs{S^c}}$, the target has Schmidt rank $d^{N-1}$ with all $q_a=d^{-(N-1)}$. Each source $e_j$ has at least one endpoint on each side (since $\abs{S},\abs{S^c}\geq 2$), so across the cut it has rank $d$ with uniform coefficients $1/d$; the $N$ independent sources give $\ket{\phi}$ rank $d^N$ with all $p_a=d^{-N}$. Both spectra are uniform and the GHZ one is supported on $d^{N-1}$ values, so Lemma~\ref{lem:LU} yields
		\begin{equation}
			\max_{U_S,U_{S^c}}\abs{\bra{\ghz{N}{d^{N-1}}}(U_S\otimes U_{S^c})\ket{\phi}} = \sum_{a=1}^{d^{N-1}}\sqrt{d^{-(N-1)}\,d^{-N}} = d^{-1/2}.
		\end{equation}
		With \eqref{eq:relax}, $\F_{K_N^{(N-1)}}\leq(d^{-1/2})^2=1/d$. Together with achievability, $\F_{K_N^{(N-1)}}=1/d$, attained by $U_{\mathrm{IS}}$. \qed
		
		\medskip\noindent\emph{Remarks.} (i) The relaxation is tight: the bipartite bound coincides with the value reached by the per-node construction, so $U_{\mathrm{IS}}$ is optimal. (ii) The argument needs $N\geq 4$ (a cut with at least two parties per side), which is the multipartite regime $k\geq 3$; the case $N=3$ ($k=2$) is the bipartite baseline already settled by~\cite{25Otfried}. (iii) Optimality is established within local unitaries on maximally entangled sources; general local channels enlarge the feasible set and can only raise $\F_G$~\cite{25Otfried}, so whether they exceed $1/d$ is left open.
	\end{appendices}
	
	\section*{Declarations}
	\subsection*{Ethics approval and consent to participate}
	Not applicable.
	\subsection*{Funding}
	Not applicable.
	
	\subsection*{Acknowledgments}
	This work was supported by the project LUQCIA Funded by the European Union – Next Generation EU, with the collaboration of the Department of Media, Connectivity and Digital Policy of the Luxembourgish Government in the framework of the RRF program.
	
	\subsection*{Code availability}
	All data and Figures are original and generated by the authors using appropriate licensed software. Scripts and routines used to produce these figures are available from the corresponding author upon reasonable request.
	
	\subsection*{Authors' contribution}
	LO and ST contributed to the idea. LO and ST developed the theory and wrote the manuscript. SC improved the manuscript and supervised the research. All authors contributed to the analysis and discussion of the results and improved the manuscript. During the preparation of this manuscript, the authors used Generative AI to improve the text's spelling, grammar, and readability. After using this tool, all authors reviewed and edited the content as needed and took full responsibility for the final version of the manuscript.

	\subsection*{Competing interests}
	The authors declare no competing financial or non-financial interests.
	
	\bibliography{ref_v2.bib}
	
\end{document}